\documentclass[10pt, a4paper, conference]{IEEEtran}
\usepackage{latexsym}
\usepackage{epsfig}
\usepackage{amsmath}
\usepackage{amssymb}
\usepackage{graphicx}

\newtheorem{thm}{Theorem}

\newtheorem{prop}{Proposition}
\newtheorem{problem}{Problem}

\newtheorem{defin}{Definition}
\newtheorem{ex}{Example}

\newcommand{\gd}[1]{{\mathcal{CN}}\left( #1 \right)}
\newcommand{\cov}[1]{{\mathbf{cov}}\left( #1 \right)}
\newcommand{\vect}[1]{{\mathbf{vec}}\left( #1 \right)}

\newcommand{\Ub}{\bar{U}}
\newcommand{\Ut}{\tilde{U}}
\newcommand{\Sigmab}{\bar{\Sigma}}
\newcommand{\Sigmat}{\tilde{\Sigma}}

\newcommand{\sigmat}{\tilde{\sigma}}

\newcommand{\Hh}{\hat{H}}
\newcommand{\Ht}{\tilde{H}}

\newcommand{\Ch}{\hat{C}}
\newcommand{\Ct}{\tilde{C}}

\newcommand{\C}{\mathcal{C}}
\newcommand{\E}[1]{{\mathbf{E}}\left\{ #1 \right\}}
\newcommand{\tr}[1]{{\mathbf{Tr}}\left( #1 \right)}

\newcommand{\R}{\mathbb{R}}

\def\QED{~\rule[-1pt]{6pt}{6pt}\par\medskip}
\newenvironment{proof}{{\bf Proof.\ }}{ \hfill \QED}

\title{Optimal Data and Training Symbol Ratio for Communication over Uncertain Channels}
\date{}

\author{
\IEEEauthorblockN{Ather Gattami}
\IEEEauthorblockA{Ericsson Research\\ Stockholm, Sweden\\
Email: ather.gattami@ericsson.com}
}

\begin{document}

\maketitle

\begin{abstract}
We consider the problem of determining the power ratio between the training symbols and data symbols in order to maximize the channel capacity for transmission over uncertain channels with a channel estimate available at both the transmitter and receiver.
The receiver makes an estimate of the channel by using a known sequence of training symbols. This channel estimate is then transmitted back to the transmitter. The capacity that the transceiver maximizes is the worst case capacity, in the sense that given a noise covariance, the transceiver maximizes the minimal capacity over all distributions of the measurement noise under a fixed covariance matrix known at both the transmitter and receiver. We give an exact expression of the channel capacity as a function of the channel covariance matrix, and the number of training symbols used during a coherence time interval. This expression determines the number of training symbols that need to be used by finding the optimal integer number of training symbols that maximize the channel capacity. As a bi-product, we show that linear filters are optimal at both the transmitter and receiver.
\end{abstract}


\section*{Notation}

\begin{tabular}{ll}
$\det(A)$ 		& $\det(A) = \prod_{i} \lambda_i$, where  $\{\lambda_i\}$\\
			& are the eigenvalues of the square matrix $A$.\\
$\otimes$	& $A\otimes B$ denotes the kronecker product between\\
			&  the matrices $A$ and $B$.\\
$\mathbf{E}\{\cdot\}$ 	& $\mathbf{E} \{x\}$ denotes the expected value of the\\ 
			& stochastic variable $x$.\\
$\mathbf{E}\{\cdot|\cdot\}$ 	& $\mathbf{E} \{x|y\}$ denotes the expected value of the\\
			& stochastic variable $x$ given $y$.\\
$\mathbf{cov}$ & $\mathbf{cov}\{x,y\} = \E{xy^* }$.\\
$h(x)$ 		& Denotes the entropy of $x$.\\
$h(x|y)$ 		& Denotes the entropy of $x$ given $y$.\\
$\textup{I}(x;y)$ 		& Denotes the mutual information between\\
			& $x$ and $y$.\\

$\mathcal{N}(m,V)$  & Denotes the set of Gaussian variables with\\ 
				& mean $m$ and covariance $V$.\\
\end{tabular}


\section{Introduction}

\subsection{Background}

This work considers the problem of determining the power ratio between the training symbols and data symbols in order to maximize the channel capacity for transmission over uncertain channels with channel state information at the transmitter and receiver. While the problem of MIMO communication over a channel known at the transmitter and receiver is well understood, the problem of uncertain channels still needs a deeper understanding of how much power we should spend on estimating the channel in order to transmit as much data as possible.
This problem poses a trade-off between the power ratio allocated for training and data transmission. On one hand, if we spend more power on training and less on data transmission, the data throughput will be small of course. On the other hand, if we spend most of the power on data transmission and much less on training, the channel estimate will be bad and therefore, the channel estimation error noise will be large, causing a rather low data rate. Given \textit{per symbol} power constraints, one might need  a number of training symbols in order to achieve a certain quality of channel state information, which leaves a smaller number of symbols for data transmission. This is an important constraint that we take into account in this work.

\subsection{Previous Work}
There has been a lot of work on MIMO communication in the context of uncertain channel state 
information. The seminal paper of Telatar \cite{telatar:1999} studied the problem of communication over uncertain Gaussian channels under the assumption that the channel realization is available at the receiver but not the transmitter. However, in practice, the transmitter and receiver don't have full knowledge of the channel. The work was extended in \cite{marzetta:1999} for the case of slowly varying channels. The channel estimation extension of \cite{telatar:1999} was presented in \cite{hassibi:2003}, where the problem of power ratio between the training and data symbols was studied under \textit{average} and \textit{per symbol} power constraints over the channel coherence time(the time where the channel is roughly constant). The crucial assumption of average power constraint allows for spending only one symbol on training, since the power is only limited by the total power resource available during the coherence time. For the case of \textit{per symbol} power constraints, the power ratio is harder to compute. The case where a channel estimate is available at the transmitter and reciever was studied in \cite{caire:2010} under the specific signaling strategy of zero-forcing linear beamforming and \textit{given Gaussian measurement noise}. Calculation of the channel capacity with respect to 
a channel estimate available at the receiver only was given in \cite{fodor:2014}.

\subsection{Contribution}
We consider communication over uncertain channels with channel state feedback at the transmitter. The receiver makes an estimate of the channel by using a known sequence of training symbols. This channel estimate is then \textit{transmitted back} to the transmitter. The capacity that the transceiver maximizes is the worst case capacity, in the sense that given a noise covariance, the transceiver maximizes the minimal capacity over all distributions of the measurement noise under a fixed covariance matrix known at both the transmitter and receiver. The channel estimation error implies that the covariance of the transmitted symbols over time affects both the covariance of the transmitted information symbols and the total noise covariance, which makes the signal structure more complicated, where it's not clear if the symbols should be uncorrelated over time. 
For the single input multiple output (SIMO) channel, we give an exact expression of the channel capacity as a function of the channel covariance matrix, the noise covariance matrix, and the number of training symbols used during a coherence time interval. This expression determines the number of training symbols that need to be used by finding the optimal integer number of training symbols that maximize the channel capacity. Numerical examples illustrate the trade-off between the number of training and data symbols. The results indicate that when the transmission power (or equivalently the signal to noise ratio) is high, a smaller number of training symbols is required to maximize the capacity compared to the low transmission power case.  We confirm these observations theoretically considering the asymptotic behavior of the power as it grows large or decreases to very small values.

\section{Preliminaries}

\begin{defin}[Kronecker Product]
For two matrices $A\in \R^{m\times n}$ and $B\in \R^{p\times q}$, the Kronecker product is defined as
$$
A\otimes B = 
		\begin{pmatrix}
			a_{11} B 		& a_{12} B 			& \cdots 	& a_{1n} B\\
			a_{21} B		& a_{22} B 			& \cdots 	& a_{2n} B\\
			\vdots 		& \vdots 				& \ddots  	& \vdots    \\
			a_{m1} B		& a_{m2} B 			& \cdots	& a_{mn} B
		\end{pmatrix}
$$
\end{defin}

\begin{prop}[Mutliplication Property of the Kronecker Product]
\label{KP1}
For any set of matrices $A, B, C, D$, we have that
$$
(A\otimes B)(C\otimes D)=AC\otimes BD
$$
\end{prop}
\begin{proof}
Consult \cite{horn:johnson2}.
\end{proof}

\begin{prop}[Determinant Property]
\label{detab}
For any two matrices $A\in \R^{m\times n}$ and $B\in \R^{n\times m}$, we have that
$$
\det(AB+I_m) = \det(BA+I_n)
$$
\end{prop}
\begin{proof}
Consult \cite{horn:johnson1}.
\end{proof}

\begin{prop}[AM-GM Inequality]
\label{amgm}
Let $a_{1}, a_2, ..., a_{n}$ be $n$ nonnegative real numbers. Then
$$
\frac{1}{n}\sum_{i=1}^n a_i \geq \sqrt[n]{\prod_{i=1}^n a_i}
$$ 
with equality if and only if $a_{1}= a_2 = \cdots = a_{n}$.
\end{prop}
\begin{proof}
The proof may be found in most standard  textbooks on Calculus.
\end{proof}



\section{Problem Formulation}
Let $H$ be a random channel such that $\vect{H} \sim \gd{0, C}$. The realization of the channel is assumed to be constant over
a coherence time interval corresponding to a block of $T$ symbols. Let $x(t)\sim \gd{0, X}$ be the transmitted symbol over the Gaussian channel at time $t$. The received signal at time $t$ is given by
$$
y(t) = H x(t) + w(t)
$$
where $\{w(t)\}$ is random white noise process, independent of $H$ and $x$, with zero mean and covariance given by
$\E{w(t) w^* (t)} = I_m$ known at both the transmitter and receiver. Without loss of generality, we assume that $W$ is invertible.

Let $(x(1), ..., x(T))$ be the block symbols transmitted within the channel coherence time $T$.
The \textit{average} power constraint imposed on the \textit{block} is given by
$$
\frac{1}{T}\sum_{t=1}^T \E{|x(t)|^2} \leq P
$$ 
The power constraint above could allow for some symbols $x(t)$ to have a larger power than $P$. However, in the real world, we have hard constraints on the peak average power per symbol, so we will impose power constraints per symbol given by
$$
\E{|x(t)|^2} \leq P, ~~~~ \text{for }~ t = 1, ..., T
$$ 
Suppose that we allocate $T_\tau$ training symbols for the channel estimation part and $T_d = T-T_\tau$ symbols  for data transmission. If a performance criterion is measured over the total channel coherence time $T$, then clearly the optimal strategy is to transmit the training symbol sequence first followed by the data symbols. Also, when no energy constraints are present, the optimal power allocation is for each symbol to be transmitted with full power $P$. However, it's not clear what the optimal \textit{ratio} between the training and data symbols. That is, what is the optimal choice of $T_\tau$ to maximize the channel capacity? To this end, we will derive the exact expression of the channel capacity SIMO channel. 
 
\subsection{SIMO Channel Estimation}
Consider a random Gaussian channel $H$ taking values in $\R^m$.
Let $x_\tau$ be a deterministic training symbol known at the transmitter and receiver with $|x_\tau|^2 = P$. 
The transmitted training sequence is given by
$$
x(t) = x_\tau, ~~~~ \text{for }~ t = 1, ..., T_\tau
$$
At the receiver, we obtain the measurement symbols
$$
y(t) = H x_\tau + w(t), ~~~~ \text{for }~ t = 1, ..., T_\tau
$$ 
Introduce the vectors
$$
y_\tau = \begin{pmatrix}
			y(1)\\
			y(2)\\
			\vdots \\
			y(T_\tau)
		\end{pmatrix}, ~~~~
w_\tau = \begin{pmatrix}
			w(1)\\
			w(2)\\
			\vdots \\
			w(T_\tau)
		\end{pmatrix}		
$$
and let the covariance matrix of $w_\tau$ be 
$$
\E{w_\tau w_\tau^*} = I_m\otimes I_{T_\tau} = I_{mT_\tau}
%
$$
It's well known that the optimal estimator $\Hh$ of $H$ given $y_\tau$ that minimizes the MSE is given
by   

\begin{align}
\Hh &= \E{H ~|~ y_\tau} \label{Hhat} \\	
	&= \cov{H, y_\tau}\cdot \cov{y_\tau, y_\tau}^{-1}y_\tau \\
	&= \E{H y^*_\tau} \cdot \E{y_\tau y_\tau^*}^{-1} y_\tau \\
	&= x^*_\tau C \left(P C + \frac{1}{T_\tau} I_m \right)^{-1}\frac{1}{T_\tau}\sum_{t=1}^{T_\tau} y(t) \\
	&= x^*_\tau C \left(P C + \frac{1}{T_\tau} I_m \right)^{-1} (Hx_\tau + \bar{w}(t))
\end{align}
where 

$$
\bar{w}(t) = \frac{1}{T_\tau}\sum_{t=1}^{T_\tau} w(t)
$$
The covariance of $\bar{w}(t)$ is easily obtained from the expression above, and it's given by 
$$
\E{\bar{w}(t)\bar{w}^*(t)} = \frac{1}{T_\tau} I_m
$$
The channel estimation error is given by
\begin{align}
\Ht 	&= H - \Hh \\		
	&= H -  x^*_\tau C \left(P C + \frac{1}{T_\tau} I_m \right)^{-1} (Hx_\tau + \bar{w}(t)) \\
	&= (I_m - P C \left(P C + \frac{1}{T_\tau} I_m \right)^{-1})H \nonumber\\
	& ~~~	- x^*_\tau C \left(P C + \frac{1}{T_\tau} I_m \right)^{-1} \bar{w}(t) \\
	&= \frac{1}{T_\tau}\left(P C + \frac{1}{T_\tau} I_m \right)^{-1} H \nonumber \\
	& ~~~ - x^*_\tau C \left(P C + \frac{1}{T_\tau} I_m \right)^{-1} \bar{w}(t)
\end{align}
Now we have that

\begin{align}
\Ch &= \E{\Ch} \label{Chat} \\
	&= P C \left(P C + \frac{1}{T_\tau} I_m \right)^{-1}C 
\end{align}
It's well known that $\Hh$ and $\Ht$ are independent since $H$ and $y_\tau$ are jointly Gaussian. Thus,

\begin{align}
\Ct 	&= \E{\Ht \Ht^*}\\
	&= \E{(H - \Hh)(H - \Hh)^*}\\
	&= C - \Ch\\
	&= C -  P C \left(P C + \frac{1}{T_\tau} I_m \right)^{-1}C\\
	&= (I_m -  P C \left(P C + \frac{1}{T_\tau} I_m \right)^{-1})C\\
	&= \frac{1}{T_\tau}\left(P C + \frac{1}{T_\tau} I_m \right)^{-1}C \label{Ctilde}
\end{align}

\subsection{Channel Capacity}
In this section, we will derive a formula for the channel capacity under the assumption that the transmitted and receiver have a common estimate of the channel. The estimate is the optimal estimate that minimizes the expected value of the variance of the estimation error.

Suppose that we transmit $T_\tau$ training symbols during the time interval $t=1, ..., T_\tau$ and $T_d = T-T_\tau$ data symbols during the time interval $t=T_\tau+1, ..., T$.  The received noisy measurements of the data symbols are given by
\begin{align}
y(t) &= H x_d(t) + w(t)\\
	&= \Hh x_d(t) + \Ht x_d(t) + w(t)\\
	&= \Hh x_d(t) + v(t), 
\end{align}
for $ t=T_\tau+1, ..., T$. Note that $v(t) = \Ht x_d(t) + w(t)$ is uncorrelated with $x_d(t)$ and $\Hh x_d(t)$ jointly. 
Introduce the vectors
$$
y_d = \begin{pmatrix}
			y(T_\tau+1)\\
			y(T_\tau+2)\\
			\vdots \\
			y(T)
		\end{pmatrix}, ~~~~
x_d = \begin{pmatrix}
			x_d(T_\tau+1)\\
			x_d(T_\tau+2)\\
			\vdots \\
			x_d(T)
		\end{pmatrix}		
$$
$$
 w_d = \begin{pmatrix}
			w(T_\tau+1)\\
			w(T_\tau+2)\\
			\vdots \\
			w(T)
		\end{pmatrix}, ~~~~
v_d = \begin{pmatrix}
			v(T_\tau+1)\\
			v(T_\tau+2)\\
			\vdots \\
			v(T)
		\end{pmatrix}				
$$
and let the resepective covariance matrices of $ w_d$ and $v_d$ be 
$$W_d=\E{ w_d w_d^*} = I_m \otimes I_{T_d} = I_{mT_d}$$
and
$$V_d = \E{v_d v_d^*} = V\otimes I_{T_d} =
		\begin{pmatrix}
			V 		& 0 			& \cdots 	& 0\\
			0		& V 		& \cdots 	& 0\\
			\vdots 	& \vdots 	& \ddots      	& \vdots \\
			0		& 0 			& \cdots		& V
		\end{pmatrix}
$$
The \textit{worst case} capacity of the channel with repsect to the measurement noise $ I_{mT_d}$ is given by 
\begin{equation}
\label{wcost}
	\C(T_{\tau}) =  \sup_{\substack{ x_d \\ \E{|x_d(t)|^2}= P}}  \inf_{\substack{ v_d\\ \E{ v_d  v_d^*}= I_{mT_d}}} \textup{I}(x_d; y_d)
\end{equation}
Thus, the problem that we want to solve is as follows.

\begin{problem}

Find the optimal integer $T_\tau \in [1, T]$ such that

\begin{equation*}
	 \C(T_{\tau}) = \sup_{\substack{ x_d \\ \E{|x_d(t)|^2}= P}} \inf_{\substack{ v_d\\ \E{ v_d v_d^*}= I_{mT_d}}}  
	 \textup{I}(x_d; y_d)
\end{equation*}
is maximized.

\end{problem}    

\section{Main Results}

It's well known that for a deterministic channel $\Hh_d$ and signal measurement 
$$
y_d = \Hh_d x_d + v_d
$$
with measurement noise $v_d$ uncorrelated with the transmitted signal $x_d$, the optimal transmitting strategy is for $x_d$
 to be Gaussian in order to minimize the worst case noise $v_d$ which is also shown to be Gaussian. In other words, the Gaussian input and noise form a Nash equilibrium. The case when the noise $v_d$ has an \textit{arbitrary} covariance matrix 
 $V_d$ has been solved in \cite{hassibi:2003}. In our case, the situation is different. The covariance matrix of $v_d$ has a \textit{structure} that also depends on the choice of the covariance matrix of the transmitted signal $x_d$. More precisely, recall that we have assumed that $\{w(t)\}$ is a temporally uncorrelated noise process with arbitrary distribution. Introduce
$$\Hh_d = \Hh \otimes I_{T_d} =
		\begin{pmatrix}
			\Hh 	& 0 			& \cdots 	& 0\\
			0		& \Hh 		& \cdots 	& 0\\
			\vdots 	& \vdots 	& \ddots      	& \vdots \\
			0		& 0 			& \cdots		& \Hh
		\end{pmatrix}
$$
and
$$\Ht_d = \Ht \otimes I_{T_d} =
		\begin{pmatrix}
			\Ht 		& 0 			& \cdots 	& 0\\
			0		& \Ht 		& \cdots 	& 0\\
			\vdots 	& \vdots 	& \ddots      	& \vdots \\
			0		& 0 			& \cdots		& \Ht
		\end{pmatrix}
$$
We may write $ w_d = v_d - \Ht_d x_d$, with $v_d$ uncorrelated with $\Hh_d$ and $x_d$. Furthermore, 
$$
\E{v_d v_d^*} - \E{\Ht_d x_d x_d^*\Ht_d^*}  = \E{ w_d w_d^*} = I_{mT_d}
$$
Let $X_d = \E{x_d x_d^*}$. Then we have
$$
\Hh_d x_d = 		\begin{pmatrix} \Hh x_d(T_\tau +1) \\ \Hh x_d(T_\tau +2)\\ \vdots \\ 
\Hh x_d(T) \end{pmatrix} 
$$
and the blocks of $\Hh_d x_d x_d^* \Hh_d^*$ at position $(i,j)$ will be given by(since $x_d(T_\tau +i)$ is a scalar for all $i = 1, ..., T-T_\tau$)
\begin{equation*}
\begin{aligned}
\\
[ \Hh_d x_d x_d^* \Hh_d^* ]_{ij} 	&= \Hh x_d(T_\tau +i) x_d^*(T_\tau + j)\Hh^*\\
									&=x_d(T_\tau +i) x_d^*(T_\tau + j)\Hh\Hh^* ,
\end{aligned}
\end{equation*}
for $i, j = 1, ..., T - T_\tau$.
Thus, $\Hh_d x_d x_d^* \Hh_d^* = (x_d x_d^*)\otimes (\Hh \Hh^*)$ and given $\Hh$, we get
$$\E{\Ht_d x_d x_d^*\Ht_d^*} = X_d \otimes \Ch.$$
Similarly, we get $\Ht_d x_d x_d^* \Ht_d^* = (x_d x_d^*)\otimes (\Ht \Ht^*)$ and 
$$\E{\Ht_d x_d x_d^*\Ht_d^*} = X_d \otimes \Ct.$$ This gives in turn
$$
V_d = \E{v_d v_d^*} =  I_{mT_d} + X_d \otimes \Ct
$$
In particular, the block diagonal elements of $V_d$ are equal, with the block elements given by
$$
V = I_m + \E{\Ht P\Ht^*}= I_m + P\Ct
$$
Now the cost given by (\ref{wcost}) may be written in terms of $v_d$ instead. That is,
\begin{equation}
\label{vcost}
	\C(T_{\tau}) =  \sup_{\E{|x_d(t)|^2}= P} \inf_{\substack{v_d\\ \E{v_d v_d^*}= I_{mT_d} + X_d \otimes \Ct}} \textup{I}(x_d; y_d)
\end{equation}

\begin{thm}
\label{thm1}
Consider a communciation channel given by $y(t) = Hx(t) + w(t)$ with $H$ taking values in $\R^m$ 
and $H \sim \gd{0, C}$, $t = 1, ..., T-T_\tau$. Let $\Hh$ be the channel estimate that is available at both the transmitter and receiver, based on $T_\tau$ training symbols. 
Under the power constraint $\E{x^2(t)}\leq P$, 
the worst case capacity  $\C(T_{\tau})$ is given by
\[
\begin{aligned}
	\C(T_{\tau})  	&= (T-T_\tau) (\log_2{\det(PC + I_m)}- \log_2{\det(P\Ct + I_m)})
\end{aligned}
\]
with
$$
\Ct = \frac{1}{T_\tau}\left(P C + \frac{1}{T_\tau} I_m \right)^{-1}C
$$
\end{thm}
Furthermore, the optimal receiver is given by the linear estimator $\E{x_d | \Hh x_d + v_d}$ 
with $v_d\sim \gd{0, P\Ct}$.

\begin{proof} 
The proof is deferred to the appendix.
\end{proof}

Theorem \ref{thm1} can be checked for the extreme cases of $T_{\tau} = 0$
and $T_{\tau} = T$. For the case $T_{\tau} = 0$, clearly no channel estimation is obtained and 
the expression of $\Ct$ reveals that this error becomes infinite. Thus, the capacity becomes $-\infty$ and no information may be transmitted. The other extreme, $T_{\tau} = T$, means that
all power is spent on channel estimation and no data may be transmitted. We see that the expression of the channel capacity formula gives zero capacity, agreeing with the physical model.

\section{Numerical Examples}

\begin{figure}  
\begin{center}  
 \includegraphics[height=10.5cm, trim=30mm 60mm 30mm 50mm, clip]{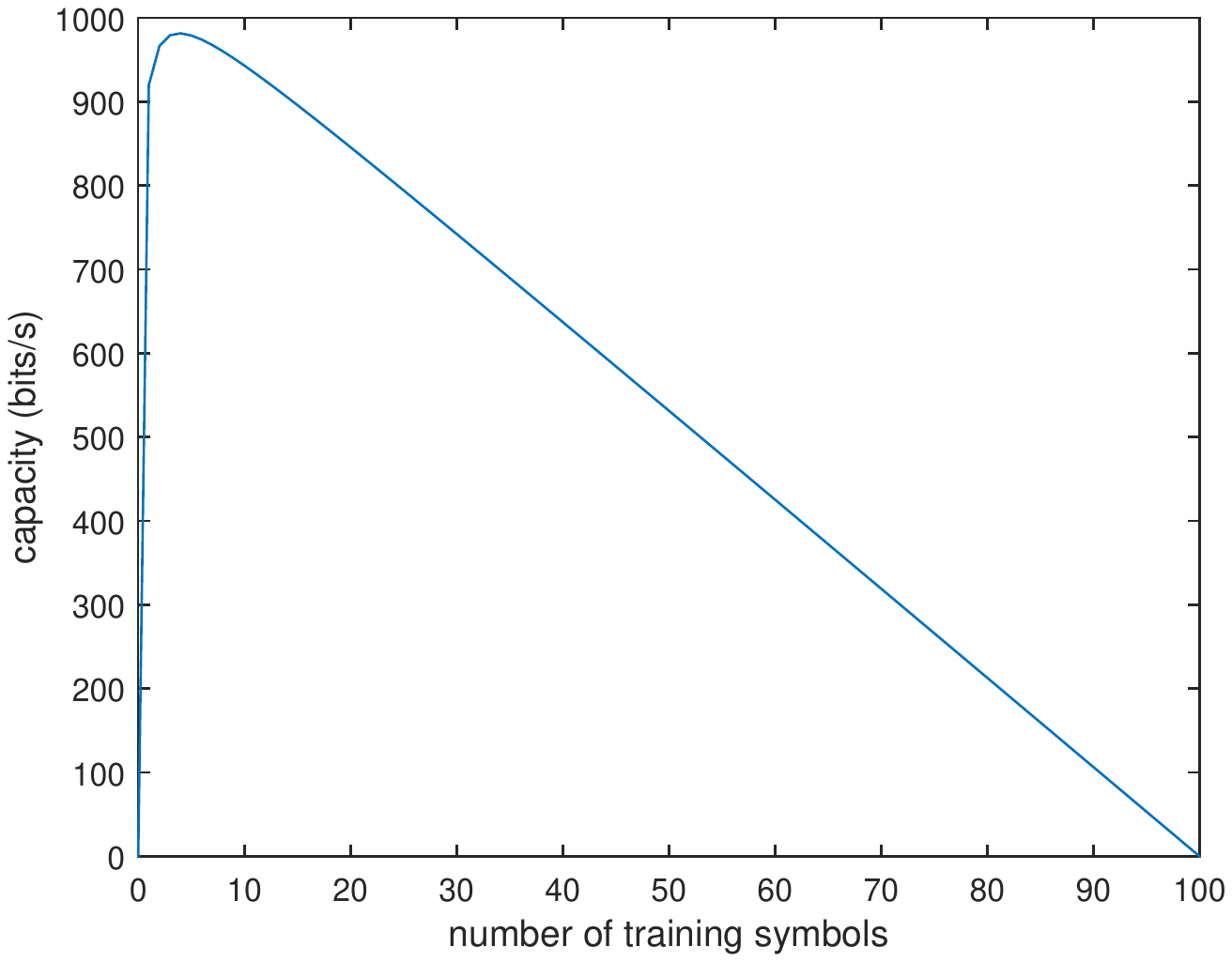}  
\caption{\small \sl A plot of the $2\times 1$ channel capacity as a function of the number of training symbols with tranmission power given by $P = 100$. \label{fig:highP}}  
\end{center}  
\end{figure}

\begin{figure}  
\begin{center}  
 \includegraphics[height=10.5cm, trim=30mm 60mm 30mm 50mm, clip]{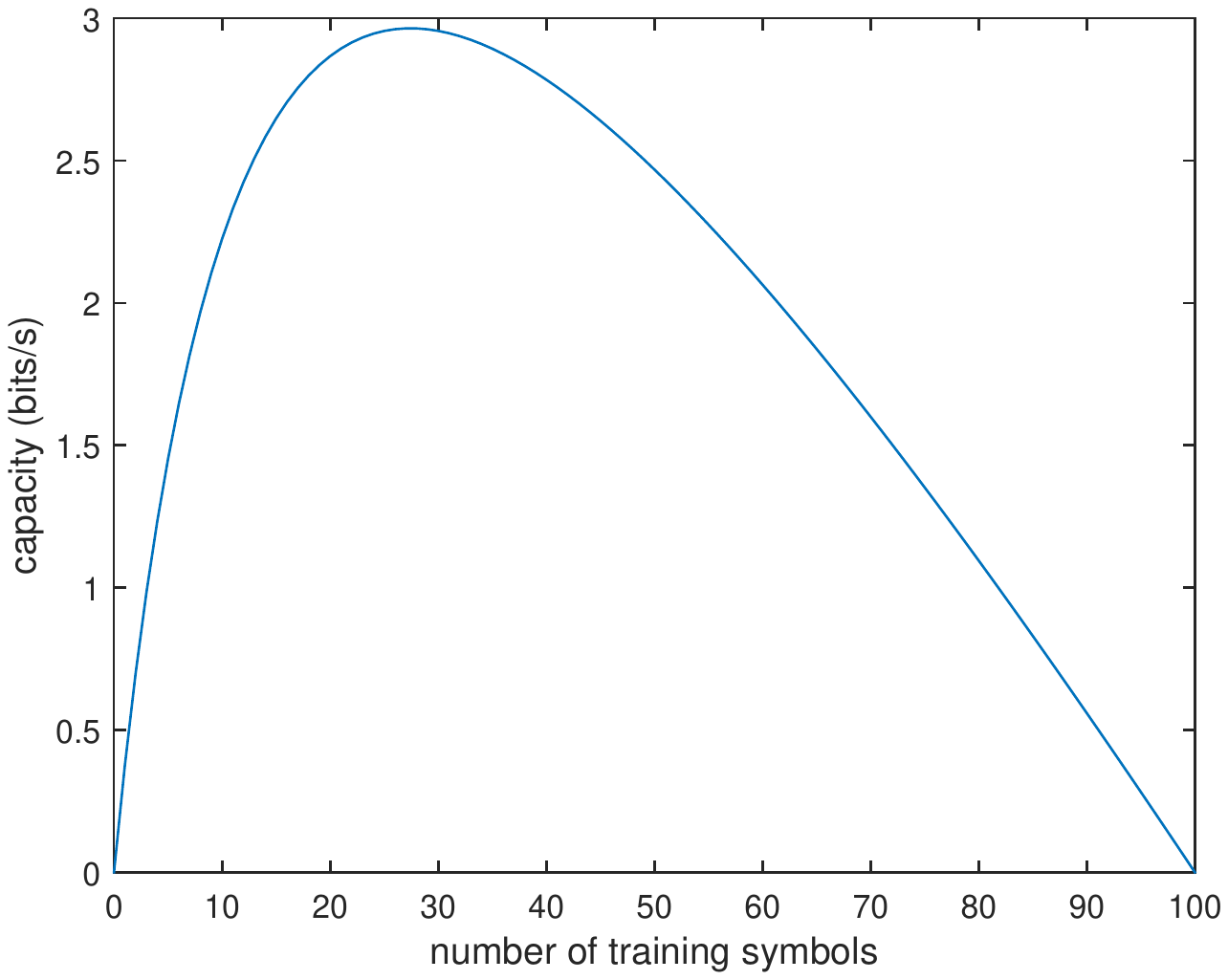}  
\caption{\small \sl A plot of the $2\times 1$ channel capacity as a function of the number of training symbols with tranmission power given by $P = 0.01$. \label{fig:lowP}}  
\end{center}  
\end{figure}  

\begin{ex}
\label{ex1}
Consider a communication SIMO channel $H$ with perfect feedback and a randomly generated covariance matrix
$$
C=
\begin{pmatrix}
	0.7426   & -0.7222\\
   -0.7222   & 6.4075
\end{pmatrix}
$$
The length of the block is assumed to be $T=100$. 
Figure \ref{fig:highP} shows that the optimum number of training symbols is $T_\tau=4$ when the power cosntraint is given by $P=100$, whereas for $P=0.01$, Figure \ref{fig:lowP}
shows that the capacity is maximized for $T_\tau=27$, which is 27\% of the available transmission power.
\end{ex}

\begin{ex}
\label{ex2}
Consider a communication SIMO channel $H$ with perfect feedback and a randomly generated covariance matrix $C$ given in the appendix.
The length of the block is assumed to be $T=100$. 
Figure \ref{fig:10highP} shows that the optimum number of training symbols is $T_\tau=2$ when the power cosntraint is given by $P=100$, whereas for $P=0.01$, Figure \ref{fig:10lowP}
shows that the capacity is maximized for $T_\tau=19$. The example clearly shows that the number of pilots needed is smaller when the number of receiving antennas increases. This is because of the increasement of the received signal power which makes the system less sensitive to channel estimation errors.
\end{ex}

\begin{figure}  
\begin{center}  
 \includegraphics[height=10.5cm, trim=30mm 60mm 30mm 50mm, clip]{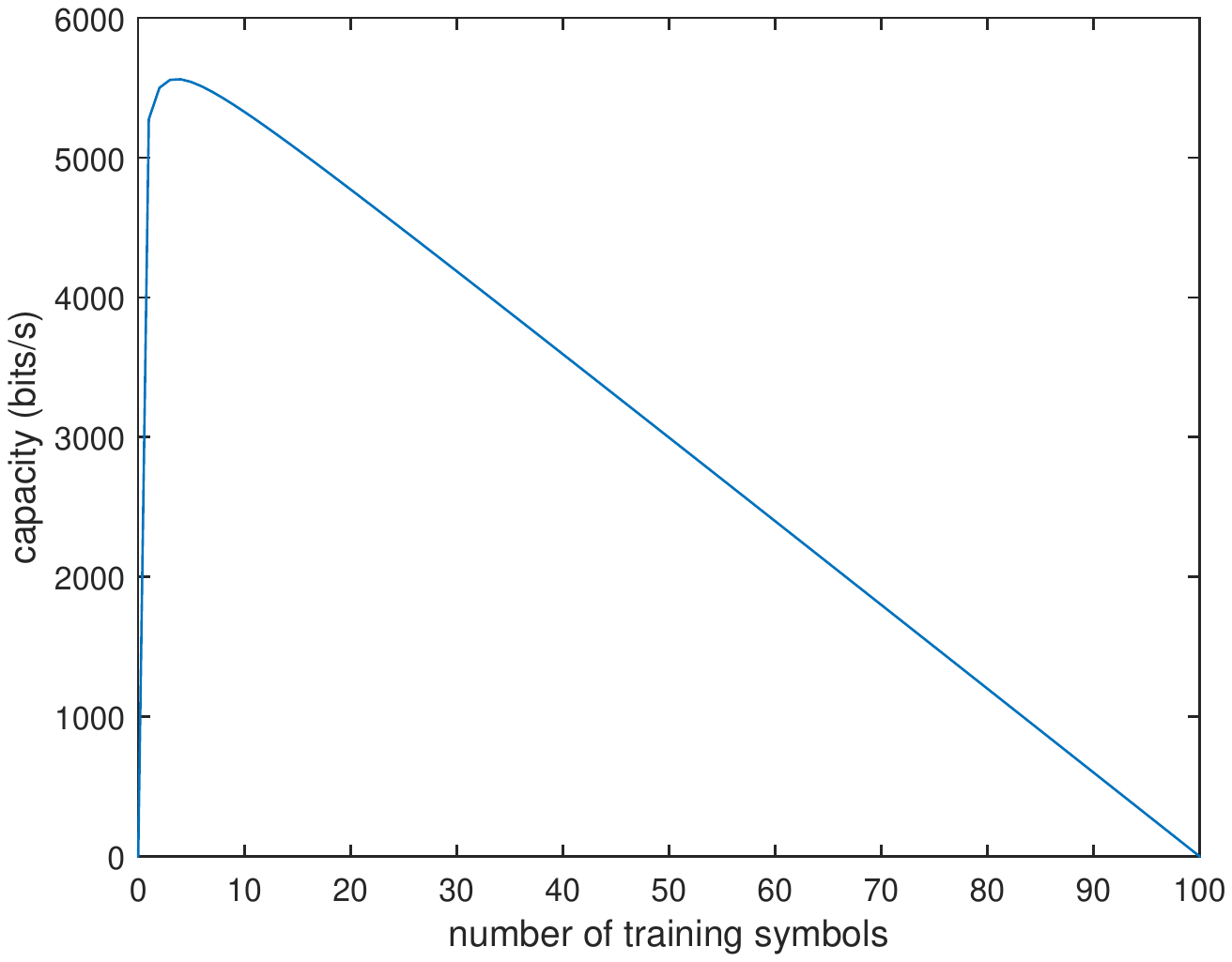}  
\caption{\small \sl A plot of the $10\times 1$ channel capacity as a function of the number of training symbols with tranmission power given by $P = 100$. \label{fig:10highP}}  
\end{center}  
\end{figure}

\begin{figure}  
\begin{center}  
\includegraphics[height=10.5cm, trim=30mm 60mm 30mm 50mm, clip]{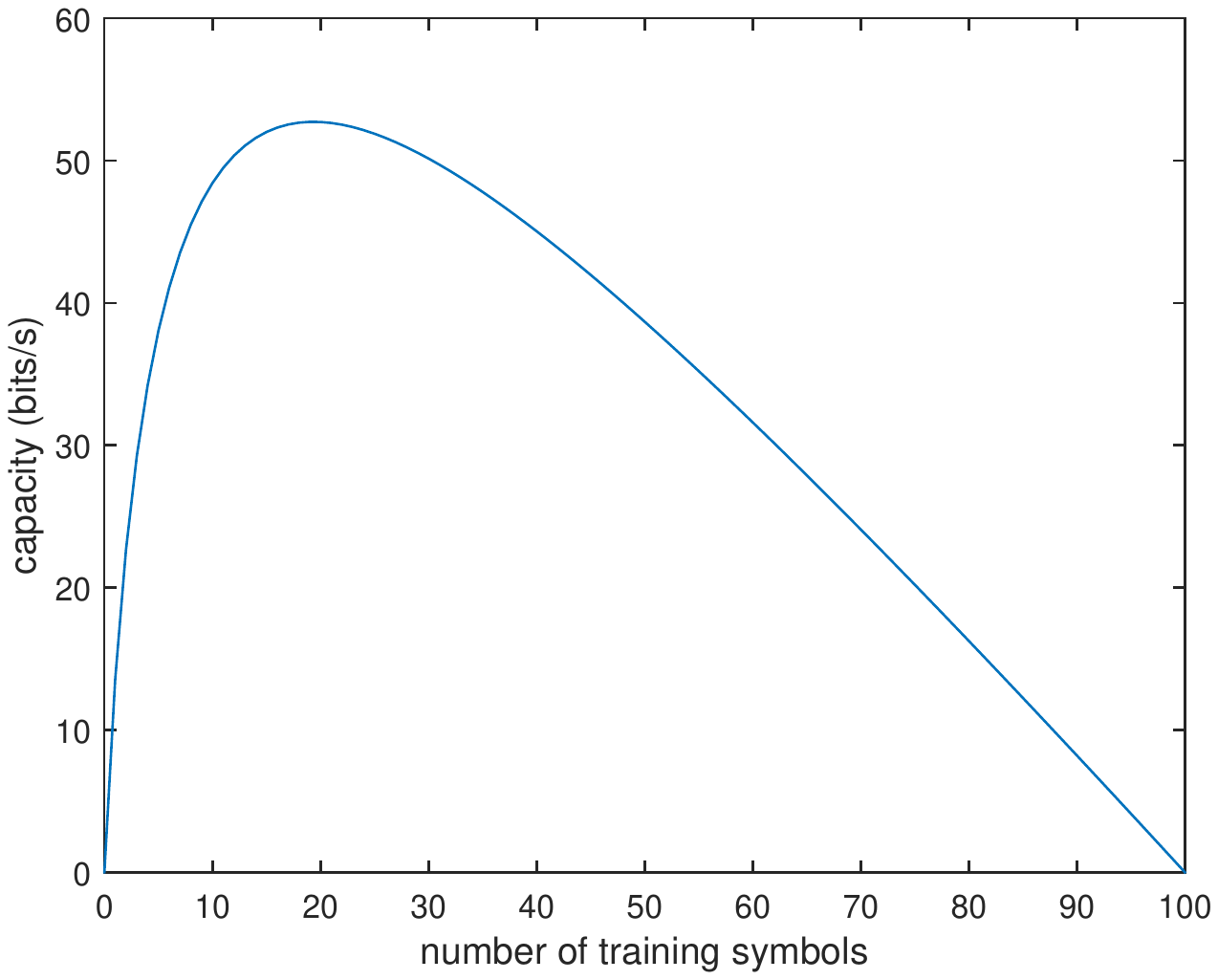}  
\caption{\small \sl A plot of the $10\times 1$ channel capacity as a function of the number of training symbols with tranmission power given by $P = 0.01$. \label{fig:10lowP}}  
\end{center}  
\end{figure}

\section{Asymptotic Results}

The previous numerical examples clearly show the the dependence of the number of training symbols on the
the symbol transmit power. We may understand this relation by looking at the asymptotic results when the 
power $P$ tends to $0$ or infinity (low respectively high signal to noise ratio). Recall that the covariance of the
channel estimation error is given by 
$$
\Ct = \frac{1}{T_\tau}\left(P C + \frac{1}{T_\tau} I_m \right)^{-1}C
$$
Cleary, if $P$ is very small and $T_\tau$ is small, then $ \frac{1}{T_\tau} I_m \succ PC$, and so 
$P C + \frac{1}{T_\tau} I_m$ will be dominated by $ \frac{1}{T_\tau} I_m$. Thus
$$
\Ct \rightarrow  \frac{1}{T_\tau}\left(\frac{1}{T_\tau} I_m \right)^{-1}C = C
$$
as $P\rightarrow 0$. This implies that the capacity would tend to zero too. Therefore, for small $P$, the number of training symbols $T_\tau$ 
must be large in order for $ \frac{1}{T_\tau} I_m$ to be of the same order as $PC$ and so making $\Ct$ in some sense.

Now consider the other extreme, that is when $P$ is large and recall the capacity formula 
\[
\begin{aligned}
	\C(T_{\tau})  	&= (T-T_\tau) (\log_2{\det(PC + I_m)}- \log_2{\det(P\Ct + I_m)})
\end{aligned}
\]
In this case, we get $ \frac{1}{T_\tau} I_m \prec PC$.
Inspecting the formula for $\Ct$ again, we see that 
$P C + \frac{1}{T_\tau} I_m$ will be dominated by $PC$ and thus
$$
\Ct \approx \frac{1}{T_\tau}\left(PC \right)^{-1}C = \frac{1}{T_\tau P}.
$$
Hence,
\[
\begin{aligned}
\log_2{\det(P\Ct + I_m)} &\approx \log_2 \det \left( \frac{1}{T_\tau}I_m + I_m \right) \\
							&=  m \log_2 \left( \frac{1}{T_\tau} + 1 \right)
\end{aligned}
\]
Note also that for large $P$, we have that
$$
\log_2 \det(PC+I_m) \approx  \log_2 \det (PC). 
$$
So for large $P$, the capacity may be well approximated by
\[
\begin{aligned}
	\C(T_{\tau})  	&\approx  (T-T_\tau) (\log_2{\det(PC)}- m \log_2 \left( \frac{1}{T_\tau} + 1 \right))
\end{aligned}
\]
We see that the capacity increases linearly with decreasing $T_\tau$ while it increases logarithmically with increasing
$T_\tau$. Therefore, the optimal choice of the number of training symbols $T_\tau\rightarrow 1$ as $P\rightarrow \infty$.

\section{Conclusion}
We considered the problem of deciding the power ratio between the training symbols and data symbols in order to maximize the channel capacity for transmission over uncertain channels with channel state information at the transmitter and receiver.
We considered the worst case capacity as a performance measure, where the transceiver maximizes the minimal capacity over all distributions of the measurement noise with a fixed covariance known at both the transmitter and receiver. We presented an exact expression of the channel capacity as a function of the channel covariance matrix, the noise covariance matrix, and the number of training symbols used during a coherence time interval. This expression determines the number of training symbols that need to be used by finding the optimal integer number of training symbols that maximize the channel capacity.
We also showed by means of numerical examples the trade-off between the number of training and data symbols. The results indicate that when the transmission power (or equivalently the signal to noise ratio) is high, a smaller number of training symbols is required to maximize the capacity compared to the low transmission power case.  We confirm these observations theoretically considering the asymptotic behavior of the power as it grows large or decreases to very small values.
Future work considers the general MIMO case, which is more involved due to combinatorial issues arising in choosing the number of training symbols for different transmitting antennas.

\bibliographystyle{plain}
\bibliography{../ref/mybib}

\appendix


\subsection*{Proof of Theorem \ref{thm1}}
Define 

\begin{equation}
\label{vcost}
	{\C}_\star = \sup_{y_d \sim \gd{0, X_d \otimes C + W_d}} \inf_{\substack{v_d\sim \gd{0, V_d}\\ V_d= I_{mT_d} + X_d \otimes \Ct}}  \textup{I}(x_d; y_d)
\end{equation}
and
\begin{equation}
\label{vcost}
	\overline{\C} =  	\sup_{\E{x_d(t) x_d^*(t)}= X_d}
						\inf_{\substack{v_d\sim \gd{0, V_d}\\ V_d= I_{mT_d} + 
							X_d \otimes \Ct}}  \textup{I}(x_d; y_d),
\end{equation}

%

Let the estimator at the receiver be a linear function of the received signal $y_d$, that is
the estimate of $x_d$ is given by $\hat{x}_d = L_d(\Hh_d x_d+v_d)$ for some matrix $L_d\in \R^{T_d\times mT_d}$. Since the receiver could be chosen nonlinear, we will get a lower bound on the capacity $\C(T_\tau)$.
Introduce $\tilde{x}_d = x_d - \hat{x}_d$. Then
\begin{align}
\C_\star 	&= \textup{I}(x_d; y_d) \\
			&=	h(x_d) - h(x_d | y_d)\\
			&=	h(x_d) - h(\hat{x}_d + \tilde{x}_d | y_d)\\  
			&\geq  h(x_d) - h(\tilde{x}_d) \label{condind}
\end{align}
Thus, 
\begin{align}
\inf_{\substack{v_d\\ V_d= I_{mT_d} + 
							X_d \otimes \Ct}}  \textup{I}(x_d; y_d)	
						&= h(x_d) - h(\tilde{x}_d)
\end{align}
which is attained for $v_d$ such that $y_d = \Hh_d x_d+v_d$ is Gaussian, which implies in turn that
$\hat{x}_d = L_d(\Hh_d x_d+v_d)$ and $\tilde{x}_d$ are Gaussian, and $\tilde{x}_d$ is independent of $y_d$, so equality holds in (\ref{condind}). Also, $h(x_d)$ is maximized for $x_d$ when it's Gaussian under a fixed covaraince. Hence,
$
\C(T_\tau)\geq \C_\star.
$
Now consider an arbitrary receiver, that is not necessarily linear and suppose that $v_d$ is Gaussian and independent of $x_d$. This gives the capacity upperbound $\overline{\C} \geq \C(T_\tau)$. We have that
\begin{align}
\textup{I}(x_d; y_d)	&=	h(y_d) - h(y_d | x_d)\\  
						&=	h(y_d) - h(v_d | x_d)\\
						&= h(y_d) - h(v_d) \label{condindv}
\end{align}
where the inequality (\ref{condindv}) holds since $v_d$ is assumed to be independent of $x_d$.
Since $h(y_d)$ is maximized when $y_d$ is Gaussian, we get
$
\C(T_\tau)\leq \overline{\C} = \C_\star.
$
Since we already have the inequality $\C(T_\tau)\geq \C_\star$, we conclude that $\C(T_\tau) = \C_\star$, and clearly $x_d$ and $y_d$ Gaussian give the worst case capacity $\C(T_\tau) = \C_\star$.

Now let the eigenvalue decompositions of $X_d$, $C$, and $\Ct$ be given by $X_d = U\Sigma U^*$, $\Sigma = \text{diag}(\sigma_1, ..., \sigma_{T_d})$, 
$C = \Ub\Sigmab \Ub^*$, $\Sigmab = \text{diag}(\Sigmab_1, ..., \Sigmab_{T_d})$, and
$\Ct = \Ut\tilde{\Sigma} \Ut^*$, $\Sigmat = \text{diag}(\sigmat_1, ..., \sigmat_{T_d})$ 
. Then,  the mutual information between the Gaussian input $x_d$ and Gaussian output $y_d$ satisfies
\begin{align}
\textup{I}(x_d; y_d) 	
			&= h(y_d) - h(y_d | x_d) \\
			&= \log_2{\det(X_d \otimes C  + I_m\otimes I_{T_d})}\nonumber \\
			& ~~~ - \log_2{\det(X_d \otimes \Ct + I_m\otimes I_{T_d})} \\
			&= \log_2{\det(X_d \otimes C +
				 I_{mT_d})} \nonumber \\
			& ~~~ - \log_2{\det(X_d \otimes \Ct + I_{mT_d})}\\
			&= \log_2{\det((U \otimes \Ub)(\Sigma \otimes \Sigmab)(U \otimes \Ub)^* 
				+ I_{mT_d})} \nonumber \\
			& ~~~ - \log_2{\det((U \otimes \Ut)(\Sigma \otimes \Sigmat)(U \otimes \Ut)^* 
				+ I_{mT_d})} \label{KP} \\
			&= \log_2{\det((U \otimes \Ub)^*(U \otimes \Ub)(\Sigma \otimes \Sigmab) 
				+ I_{mT_d})} \nonumber \\
			& ~~~- \log_2{\det((U \otimes \Ut)^*(U \otimes \Ut)(\Sigma \otimes \Sigmat) 
				+ I_{mT_d})} \label{det} \\
			&= \log_2{\det(\Sigma \otimes \Sigmab + I_{mT_d})} \nonumber\\
			&   ~~~ - \log_2{\det(\Sigma \otimes \Sigmat + I_{mT_d})} \label{UU}\\
			&= \log_2{\left(\prod_{i=1}^{T_d} \det(\sigma_i \Sigmab+ I_m ) \right)} \nonumber \\
			&  ~~~ - \log_2{\left(\prod_{i=1}^{T_d} \det(\sigma_i \Sigmat+ I_m ) \right)}\\
			&= T_d \log_2 \left(\prod_{i=1}^{T_d} \frac{\det(\sigma_i \Sigmab+ I_m )}{\det(\sigma_i \Sigmat+ I_m )} \right)^{\frac{1}{T_d}}\\
			&\leq T_d \log_2 \left(\frac{1}{T_d}\sum_{i=1}^{T_d} \frac{\det(\sigma_i \Sigmab+ I_m )}{\det(\sigma_i \Sigmat+ I_m )} \right) \label{AG}
\end{align}
where (\ref{KP}) follows from Proposition \ref{KP1}, (\ref{det}) follows from Proposition \ref{detab}, (\ref{UU}) follows from Proposition \ref{KP1}, and  (\ref{AG}) follows from Proposition \ref{amgm}, with equality if and only if $\sigma_1 = \sigma_2 = \cdots = \sigma_{T_d}$. Since $\sigma_1 + \sigma_2 + \cdots + \sigma_{T_d}  = \tr{\Sigma} = \tr{X} = nP$, we must have $$\sigma_1 = \sigma_2 = \cdots = \sigma_{T_d} = P$$ This implies that the capacity maximizing input covariance is $X=P\cdot I_{T_d}$. Thus, the maximum capacity is given by
\begin{align}
\C(T_{\tau}) 
&= \log_2{\det((P I_{T_d})\otimes (\Ch +  \Ct)  + I_m\otimes I_{T_d})} \nonumber \\
&~~~ - \log_2{\det((PI_{T_d})\otimes \Ct + I_m\otimes I_{T_d})}\\
&= T_d (\log_2{\det(PC + I_m)} \nonumber\\
&~~~- \log_2{\det(P\Ct + I_m)})\\
&= (T-T_\tau) (\log_2{\det(PC + I_m)}\nonumber \\
&~~~- \log_2{\det(P\Ct + I_m)})
\end{align}
and the proof is complete.

\subsection*{Supplement for Example \ref{ex2}}
The random matrix $C$ used in  Example \ref{ex2} is given by
\small{
\begin{equation*}
\begin{aligned}
C &=
\left(
\begin{matrix}
12.618&-2.5315&-2.2424&1.1965&-1.5896 \\ 
-2.5315&8.7639&2.0577&-4.3889&2.0117 \\ 
-2.2424&2.0577&6.0997&0.2384&-1.1894 \\ 
1.1965&-4.3889&0.2384&7.327&1.1523 \\ 
-1.5896&2.0117&-1.1894&1.1523&10.2643 \\ 
2.5086&-0.011&-1.6082&0.3583&1.3222 \\ 
-4.5906&-0.499&-4.5764&0.705&-0.1717 \\ 
-1.398&3.1713&-2.319&-5.3612&-2.361 \\ 
1.9345&2.1956&-0.6284&-2.2747&-0.6889 \\ 
-4.0798&0.2636&-0.5846&-0.7751&1.2117
\end{matrix} \right. \\
& ~~~~~~
\left.
\begin{matrix}
2.5086&-4.5906&-1.398&1.9345&-4.0798 \\ 
-0.011&-0.499&3.1713&2.1956&0.2636 \\ 
-1.6082&-4.5764&-2.319&-0.6284&-0.5846 \\ 
0.3583&0.705&-5.3612&-2.2747&-0.7751 \\ 
1.3222&-0.1717&-2.361&-0.6889&1.2117 \\ 
2.1366&0.2868&-0.8628&1.2528&-1.1311 \\ 
0.2868&18.4323&10.9609&2.3883&2.0394 \\ 
-0.8628&10.9609&20.4969&10.5705&1.3217 \\ 
1.2528&2.3883&10.5705&9.7425&-0.307 \\ 
-1.1311&2.0394&1.3217&-0.307&3.3511
\end{matrix} \right)
\end{aligned}
\end{equation*}
}

\end{document}